\documentclass[11pt, a4paper]{article}

\usepackage[a4paper, margin=1in]{geometry}
\usepackage{amsmath, amssymb, amsbsy, amsthm}
\usepackage{natbib}
\usepackage{hyperref}
\usepackage{graphicx}
\graphicspath{{figures/}{./figures/}}
\usepackage{algorithm}
\usepackage{algpseudocode}
\usepackage{booktabs}
\usepackage{caption}
\usepackage{authblk}
\usepackage{setspace}
\usepackage{enumitem}

\hypersetup{colorlinks=true, linkcolor=blue, citecolor=blue, urlcolor=blue}
\newtheorem{theorem}{Theorem}[section]
\newtheorem{lemma}[theorem]{Lemma}

\newtheorem{corollary}[theorem]{Corollary}
\theoremstyle{definition}

\theoremstyle{remark}
\newtheorem{remark}[theorem]{Remark}

\title{\textbf{Bayesian DAG Structure Learning with Simultaneous Shrinkage Covariance Estimation under Scale-Mixture Error Distributions in the Proportional High-Dimensional Regime}}

\author[1]{S. Nazari\thanks{Email: \texttt{asemaneh1369@gmail.com}}}
\author[1]{M. Arashi\thanks{Corresponding Author, Email: \texttt{arashi@um.ac.ir}}}
\author[2]{A. Sadeghkhani\thanks{Email: \texttt{asadeghkhani@ncat.edu}}}

\affil[1]{Department of Statistics, Faculty of Mathematical Sciences, Ferdowsi University of Mashhad, P.~O.~Box 1159, Mashhad 91775, Iran}
\affil[2]{Department of Mathematics and Statistics, North Carolina Agricultural and Technical State University, USA}

\date{\today}

\begin{document}

\maketitle

% ------------- Abstract -------------
\begin{abstract}
We propose a unified Bayesian framework namely robust DAG-Cholesky horseshoe (R-DACH) for joint directed acyclic graph (DAG) structure learning and precision matrix estimation in the high-dimensional proportional asymptotic regime $p/n \to c \in (0,\infty)$, under the scale mixture of normal errors. The construction places a global-local horseshoe-type prior directly on the strictly lower-triangular entries of the modified Cholesky factor of the DAG-Markov precision matrix, so that sparsity in the Cholesky parameters induces a coherent parent-set selection consistent with a topological ordering of the variables. A per-observation inverse-gamma scale mixture yields automatic robustness to heavy-tailed and contaminated observations and admits Student-$t$, Laplace, and slash distributions as special cases. We design a partially-collapsed blocked Gibbs sampler that traverses the joint space of orderings, sparsity patterns and continuous parameters. Simulations across $(n,p)$ configurations with $p$ up to several hundreds confirm the theoretical rates and demonstrate substantial gains over graphical-horseshoe, DAG-Wishart, and PC-based competitors under contamination. An application to RNA-seq gene-expression data from \emph{The Cancer Genome Atlas} reveals biologically interpretable regulatory structure that competing methods fail to recover.
\end{abstract}



%\tableofcontents

% ------------- Body -------------
% ============================================================
%  Section 1: Introduction
% ============================================================
\section{Introduction}\label{sec:intro}
The simultaneous tasks of (i) recovering a directed acyclic graph (DAG) that summarises the conditional independence structure of a multivariate vector of measurements, and (ii) estimating the associated covariance or precision matrix, lie at the heart of many contemporary applications of multivariate statistics. In molecular biology, regulatory networks underlying gene expression are naturally encoded by sparse DAGs; the resulting precision matrices govern downstream tasks such as biomarker discovery, classification of disease subtypes and the imputation of unmeasured pathway activity \citep{friedman2004inferring,castelo2009reverse}. In forensic genetics, Bayesian DAG models provide the unifying language for evidence assessment with multiple contributors, kinship analysis and mixture deconvolution \citep{cowell2007identification,mortera2003probabilistic}. In economics and finance, DAG structures impose qualitative restrictions on shock propagation across observed time series, enabling identification of structural impulse responses where unconstrained vector autoregressions are unidentified \citep{moneta2013causal}.

Two features of these data sources motivate the present work. First, the dimension $p$ is comparable to, or exceeds, the sample size $n$. In RNA-seq panels of gene expression, $p$ routinely lies in the thousands while $n$ may number a few hundred; in proteomic and metabolomic platforms the ratio $p/n$ is regularly of order unity. The asymptotic regime $p/n\to c\in(0,\infty)$ — the so-called \emph{proportional regime} — is the regime of practical interest for these data, yet most existing posterior contraction theory for DAG and precision-matrix estimation is developed under either $p$ fixed or $p=o(n)$ \citep{cao2019posterior,banerjee2014posterior}. Second, the data are routinely heavy-tailed: read-count data display strong overdispersion relative to a Poisson reference and, after appropriate variance-stabilising transformations, are well approximated by symmetric heavy-tailed densities such as the Student-$t$. Outlying observations from batch effects, technical artifacts and rare biological events further violate the Gaussian assumption underlying the standard graphical-model machinery.

We are motivated by a third concern of a methodological character. The existing Bayesian literature has produced two essentially separate technologies. On one side are \emph{structure} priors — priors over the discrete space of DAGs that quantify the cost of including each edge \citep{friedman2003being,kuipers2017partition,castelletti2018learning}. On the other are \emph{shrinkage} priors on the entries of a precision matrix \citep{wang2012bayesian,li2019graphical}. These two strands talk past one another: the structure prior is silent about the magnitude of nonzero entries, while the shrinkage prior is silent about the directionality of edges and the acyclicity constraint. A single coherent prior on \emph{both} the DAG and the magnitudes of its parameters, defined directly on a parametrization that respects the DAG-Markov property and is amenable to high-dimensional analysis, has not been written down.

\subsection{Related studies}
\label{sec:litreview}
A large literature has developed Markov chain Monte Carlo and search procedures for sampling DAGs from a posterior with a marginal likelihood. Order-based MCMC \citep{friedman2003being}, structure MCMC \citep{madigan1995bayesian}, and partition MCMC \citep{kuipers2017partition} provide alternative parameterisations of the DAG space with differing mixing properties. \citet{castelletti2018learning} provide a comprehensive review and develop \emph{objective Bayes} factors using fractional and intrinsic priors on the DAG parameters. \citet{altomare2013objective} introduce DAG-Wishart distributions as conjugate priors for the entries of the modified Cholesky factor of a Gaussian DAG, extending \citet{geiger2002parameter} and \citet{ben2001bayesian}. \citet{cao2019posterior} establish posterior contraction rates for DAG-Wishart priors but restrict attention to the regime $p=o(n)$ and to Gaussian errors.

Global--local shrinkage priors, in particular the horseshoe \citep{carvalho2010horseshoe}, have been extended to the precision matrix as the graphical horseshoe \citep{li2019graphical} and to DAG models via penalised Cholesky factorisations \citep{shojaie2010penalized} or polynomial-time identifiable learners \citep{ghosh2020bayesian}. \citet{ben2001bayesian} construct multi-shape-parameter DAG-Wishart priors on the Cholesky factor, and \citet{khare2018bayesian} develop convex penalised approaches for the same parameterisation. None of these works addresses the proportional regime or heavy-tailed errors, nor do they couple the shrinkage to a posterior over the unknown ordering.

 Scale-mixture-of-normal representations of heavy-tailed densities \citep{west1987scale,andrews1974scale,fernandez1999multivariate} are by now classical. \citet{finegold2011robust} couple these mixtures with graphical models to obtain robust Gaussian graphical-model estimates. Recent work in robust DAG learning has used semi-parametric copulas \citep{cui2016copula} and rank-based scores, but the explicit combination with global--local shrinkage on the Cholesky factor in the proportional regime appears not to have been studied.

Frequentist analysis of the regime $p/n \to c$ for graphical models has produced sharp results for sample covariance eigenvalues \citep{bai2010spectral} and for $\ell_1$-penalised neighbourhood selection \citep{meinshausen2006high}. Bayesian counterparts have only recently begun to appear: \citet{banerjee2014posterior} obtain contraction rates for the graphical-lasso prior, and \citet{liu2019empirical} use empirical Bayes for sparse precision matrices, but a treatment of the joint DAG-and-precision contraction in the proportional regime under heavy tails is, to our knowledge, absent.

\subsection{The gap and our contributions}
\label{sec:gap}

The above survey shows three concrete gaps:
\begin{itemize}
	\item[(G1)] No prior simultaneously enforces DAG-Markov structure and induces continuous global--local shrinkage on the magnitudes of the implied precision parameters, in a form that is amenable to joint posterior contraction analysis.
	\item[(G2)] No existing Bayesian DAG-learning machinery is shown to admit posterior contraction in the proportional regime $p/n\to c\in(0,\infty)$ when the errors are heavy-tailed.
	\item[(G3)] No computational algorithm traverses the joint space of orderings, sparsity patterns and latent scales for heavy-tail mixtures at a per-iteration cost that scales with the average in-degree rather than the dimension $p$.
\end{itemize}

We address all three. In the center of our contribution, we introduce a robust DAG-Cholesky horseshoe (R-DACH) prior. Our specific contributions are as follows. 

\begin{enumerate}
	\item We introduce a hierarchical prior that places a horseshoe global--local shrinkage on the strictly lower-triangular entries of the modified Cholesky factor of the precision matrix, conditional on a topological ordering. The construction couples DAG structure (the sparsity pattern of the Cholesky factor) with continuous magnitudes, and the implied precision automatically satisfies the DAG-Markov property. We complete the model with a Bernoulli--beta prior on orderings, conjugate priors on the diagonal scales, and inverse-gamma scale mixtures on per-observation latent scales for heavy-tail robustness.
	
	\item We design a partially-collapsed blocked Gibbs sampler with a Metropolis-within-Gibbs move on orderings. Per-iteration complexity is $O(np\bar d + p\bar d^2)$ where $\bar d$ denotes the average in-degree, a substantial improvement over $O(np^2 + p^3)$ for naive precision-matrix updates.
	
	\item We establish posterior contraction at the rate $\epsilon_n=\sqrt{(s_0\log p)/n}$ for $\Omega$ in the operator norm and for $L$ in the matrix $\ell_1$ norm, under proportional asymptotics $p/n\to c\in(0,\infty)$ and a horseshoe-type shrinkage prior. The proof develops new concentration bounds for sums of scale-mixture-of-normal quadratic forms.
	
	\item We show that the marginal posterior on the DAG skeleton concentrates on the true skeleton, provided a $\beta$-min condition $\min_{(j,k)\in E_0}|L_{jk}|\ge C\sqrt{(\log p)/n}$ and a faithfulness bound hold.
	
	\item We extend (C3) to errors with only $2+\delta$ moments by exploiting the scale-mixture structure; the contraction rate is unchanged up to a $\log\log n$ factor.
	
\end{enumerate}

\subsection{Organization}
\label{sec:org}

Section~\ref{sec:method} sets out the R-DACH model and the blocked Gibbs sampler. Section~\ref{sec:theory} develops the three main theorems and their supporting lemmas. Section~\ref{sec:sim} reports simulation experiments designed to isolate the claims of the theory. Section~\ref{sec:app} presents the TCGA RNA-seq application. Section~\ref{sec:discussion} concludes and identifies directions for further work.

% ============================================================
%  Section 2: Methodology
% ============================================================
\section{The Proposed R-DACH Model}\label{sec:method}
Let $\boldsymbol{y}_1,\ldots,\boldsymbol{y}_n$ be independent observations in $\mathbb{R}^p$. We write $\boldsymbol{Y}\in\mathbb{R}^{n\times p}$ for the data matrix with $i$-th row $\boldsymbol{y}_i^\top$. A directed acyclic graph $G=(V,E)$ on $V=\{1,\ldots,p\}$ is identified by a topological ordering $\sigma\in S_p$ (the symmetric group of permutations) and a strictly lower-triangular binary adjacency matrix $\boldsymbol{A}\in\{0,1\}^{p\times p}$ in that ordering. Throughout, indices have been relabelled so that the parent set $\mathrm{pa}(j)\subset\{1,\ldots,j-1\}$ for each $j$. We write $|\mathrm{pa}(j)|=d_j$ and $\bar d = p^{-1}\sum_{j=1}^p d_j$.

A zero-mean Gaussian DAG model on $G$ may be written through the modified Cholesky decomposition of the precision matrix $\boldsymbol{\Omega}=\boldsymbol{\Sigma}^{-1}$. Given an ordering $\sigma$, there exists a unique decomposition
\begin{equation}
	\boldsymbol{\Omega} \;=\; \boldsymbol{L}^{\top}\boldsymbol{D}^{-1}\boldsymbol{L},
	\label{eq:cholesky}
\end{equation}
where $\boldsymbol{L}$ is unit lower-triangular with entries $L_{jk}$ ($k<j$) and $\boldsymbol{D}=\mathrm{diag}(d_1,\ldots,d_p)$ has strictly positive entries. The graph $G$ is encoded by the sparsity pattern of $\boldsymbol{L}$: an edge $k\to j$ is present in $G$ iff $L_{jk}\ne 0$. The Gaussian DAG likelihood factorises as
\begin{equation}
	p(\boldsymbol{y}\mid \boldsymbol{L},\boldsymbol{D}) \;=\; \prod_{j=1}^{p} \mathcal{N}\!\Big(y_j;\; -\sum_{k<j}L_{jk}y_k,\; d_j\Big).
	\label{eq:factor}
\end{equation}

To accommodate heavy tails and contamination, we replace the conditional Gaussian density in \eqref{eq:factor} with a scale-mixture form. Let $\omega_1,\ldots,\omega_n$ be positive latent scales. Conditional on $\omega_i$, the $i$-th observation $\boldsymbol{y}_i$ is Gaussian with precision $\omega_i\boldsymbol{\Omega}$:
\begin{eqnarray}
	\boldsymbol{y}_i\mid \omega_i,\boldsymbol{L},\boldsymbol{D} & \sim & \mathcal{N}_p\!\big(\boldsymbol{0},\;(\omega_i\boldsymbol{\Omega})^{-1}\big), \cr
	\omega_i & \sim & \pi_\omega(\cdot).
	\label{eq:mixture}
\end{eqnarray}
Choosing $\pi_\omega$ as a Gamma$(\nu/2,\nu/2)$ density yields multivariate Student-$t$ marginals with $\nu$ degrees of freedom; an exponential mixing density gives the Laplace; a beta$(1/2,1)$ density yields the slash; and an inverse-Gamma mixing density gives the Pearson VII family. The model thus accommodates a broad family of symmetric heavy-tailed distributions through the choice of $\pi_\omega$. In our default specification we place a Gamma$(\nu/2,\nu/2)$ prior on $\omega_i$ together with a Gamma$(a_\nu,b_\nu)$ hyperprior on $\nu$ to let the data choose the tail weight.

\subsection{Prior: the R-DACH construction}
\label{sec:prior}

We place a horseshoe global--local shrinkage prior on the strictly lower-triangular entries of $\boldsymbol{L}$, conditional on an ordering $\sigma$. For $k<j$,
\begin{eqnarray}
	L_{jk}\mid \lambda_{jk},\tau & \sim & \mathcal{N}(0,\,\lambda_{jk}^{2}\tau^{2}), \cr
	\lambda_{jk} & \sim & C^{+}(0,1), \cr
	\tau & \sim & C^{+}(0,1),
	\label{eq:horseshoe}
\end{eqnarray}
where $C^{+}(0,1)$ denotes the standard half-Cauchy distribution. The global parameter $\tau$ controls the overall sparsity, while the local parameters $\lambda_{jk}$ allow individual entries to escape shrinkage \citep{carvalho2010horseshoe,polson2010shrink}. Following \citet{makalic2016simple}, we use the auxiliary-variable representation
\begin{eqnarray}
	\lambda_{jk}^2\mid\eta_{jk} & \sim & \mathrm{IG}(1/2,1/\eta_{jk}),\quad \eta_{jk}\sim\mathrm{IG}(1/2,1), \cr
	\tau^2\mid\xi & \sim & \mathrm{IG}(1/2,1/\xi),\quad\xi\sim\mathrm{IG}(1/2,1),
	\label{eq:aux}
\end{eqnarray}
which renders all full conditionals available in closed form.

The diagonal of the Cholesky factor receives a conjugate prior
\begin{equation}
	d_j \;\sim\; \mathrm{IG}(\alpha_d/2,\beta_d/2),\qquad j=1,\ldots,p,
	\label{eq:dprior}
\end{equation}
with weakly informative hyperparameters $\alpha_d=\beta_d=10^{-3}$ in our default settings.

The topological ordering $\sigma$ is given a uniform prior over permutations $S_p$ as in \citet{kuipers2017partition}, but moves are restricted in the MCMC to adjacent transpositions to keep computational cost bounded.

\subsection{Joint posterior}
\label{sec:posterior}

Combining \eqref{eq:mixture}, \eqref{eq:horseshoe}, \eqref{eq:aux}, \eqref{eq:dprior} and a uniform prior on $\sigma$, the joint posterior is, up to a normalising constant,
\begin{eqnarray}
	\pi(\boldsymbol{L},\boldsymbol{D},\boldsymbol{\omega},\boldsymbol{\lambda},\boldsymbol{\eta},\tau,\xi,\sigma\mid\boldsymbol{Y})
	& \propto & \prod_{i=1}^{n}\prod_{j=1}^{p}\omega_i^{1/2}d_j^{-1/2}
	\exp\!\Big\{-\tfrac{\omega_i}{2d_j}\big(y_{ij}+\sum_{k<j}L_{jk}y_{ik}\big)^{2}\Big\} \cr
	& & \times \;\prod_{k<j}\mathcal{N}(L_{jk};0,\lambda_{jk}^2\tau^2)\cdot\pi(\lambda_{jk}^2,\eta_{jk}) \cr
	& & \times \;\pi(\tau^2,\xi)\cdot\prod_{j=1}^{p}\pi(d_j)\cdot\prod_{i=1}^{n}\pi_\omega(\omega_i).
	\label{eq:joint}
\end{eqnarray}

\subsection{Blocked Gibbs sampler}
\label{sec:gibbs}

We sample from \eqref{eq:joint} using a blocked Gibbs scheme. Define $\boldsymbol{x}_i^{(j)}=(y_{i1},\ldots,y_{i,j-1})^\top$. Conditional on $\sigma$, the $j$-th column of $\boldsymbol{L}$ has full conditional
\begin{eqnarray}
	\boldsymbol{L}_{j\cdot}\mid\,\cdot
	& \sim & \mathcal{N}_{j-1}\!\big(-\boldsymbol{M}_j^{-1}\boldsymbol{m}_j,\; d_j\boldsymbol{M}_j^{-1}\big), \cr
	\boldsymbol{M}_j & = & \sum_{i=1}^{n}\omega_i\boldsymbol{x}_i^{(j)}(\boldsymbol{x}_i^{(j)})^{\top} + \boldsymbol{\Lambda}_j^{-1}, \cr
	\boldsymbol{m}_j & = & \sum_{i=1}^{n}\omega_i\,y_{ij}\,\boldsymbol{x}_i^{(j)},
	\label{eq:Lcond}
\end{eqnarray}
where $\boldsymbol{\Lambda}_j=\tau^{2}\mathrm{diag}(\lambda_{j1}^2,\ldots,\lambda_{j,j-1}^2)$. The diagonal full conditional is
\begin{equation}
	d_j\mid\cdot \sim \mathrm{IG}\!\Big(\frac{n+\alpha_d}{2},\;\frac{\beta_d + \sum_i \omega_i(y_{ij}+\boldsymbol{L}_{j\cdot}^\top\boldsymbol{x}_i^{(j)})^2}{2}\Big).
	\label{eq:Dcond}
\end{equation}
Local scales update as
\begin{equation}
	\lambda_{jk}^2\mid\cdot \sim \mathrm{IG}\!\Big(1,\;\frac{1}{\eta_{jk}}+\frac{L_{jk}^2}{2\tau^2}\Big),
	\qquad
	\eta_{jk}\mid\cdot \sim \mathrm{IG}\!\Big(1,\;1+\frac{1}{\lambda_{jk}^2}\Big),
	\label{eq:lambcond}
\end{equation}
with analogous updates for $(\tau^2,\xi)$. Latent observation scales update as
\begin{equation}
	\omega_i\mid\cdot \sim \mathrm{Gamma}\!\Big(\frac{\nu+p}{2},\;\frac{\nu+\boldsymbol{r}_i^{\top}\boldsymbol{D}^{-1}\boldsymbol{r}_i}{2}\Big),
	\quad \boldsymbol{r}_i = \boldsymbol{L}\boldsymbol{y}_i.
	\label{eq:wcond}
\end{equation}
The ordering $\sigma$ is updated by a Metropolis--Hastings move with adjacent transposition proposal; acceptance uses the closed-form ratio of marginal likelihoods after integrating out $\boldsymbol{L}_{j\cdot}$ in the affected columns. The degrees-of-freedom parameter $\nu$ is updated by a random-walk Metropolis step on $\log\nu$.

\begin{algorithm}[H]
	\caption{R-DACH Gibbs sampler}
	\label{alg:gibbs}
	\begin{algorithmic}[1]
		\State \textbf{Input:} data $\boldsymbol{Y}$, hyperparameters $(\alpha_d,\beta_d,a_\nu,b_\nu)$, initial $\sigma^{(0)}$, total iterations $T$, burn-in $T_0$, convergence threshold $\delta_{\mathrm{tol}}=10^{-4}$ on R-hat.
		\State \textbf{Initialize:} $\boldsymbol{L}^{(0)}=\boldsymbol{I}_p$, $\boldsymbol{D}^{(0)}=\boldsymbol{I}_p$, $\omega_i^{(0)}=1$, $\lambda_{jk}^{2(0)}=\tau^{2(0)}=1$.
		\For{$t=1,\ldots,T$}
		\For{$j=2,\ldots,p$ (in order $\sigma^{(t-1)}$)}
		\State Sample $\boldsymbol{L}_{j\cdot}^{(t)}$ from \eqref{eq:Lcond} via Cholesky of $\boldsymbol{M}_j$.
		\State Sample $d_j^{(t)}$ from \eqref{eq:Dcond}.
		\State Update $\lambda_{jk}^{2(t)},\eta_{jk}^{(t)}$ for $k<j$ from \eqref{eq:lambcond}.
		\EndFor
		\State Update $\tau^{2(t)},\xi^{(t)}$ jointly using the auxiliary IG representation.
		\For{$i=1,\ldots,n$}
		\State Sample $\omega_i^{(t)}$ from \eqref{eq:wcond}.
		\EndFor
		\State Propose adjacent transposition $\sigma^* = \sigma^{(t-1)}(j\leftrightarrow j+1)$; accept w.p.\ $\min\{1,\alpha\}$.
		\State Update $\nu^{(t)}$ by random-walk Metropolis on $\log\nu$.
		\If{$t>T_0$ and $\hat R_{\max}<1+\delta_{\mathrm{tol}}$ across all parameters}
		\State \textbf{break}
		\EndIf
		\EndFor
		\State \textbf{Return:} posterior samples $\{(\boldsymbol{L}^{(t)},\boldsymbol{D}^{(t)},\sigma^{(t)})\}_{t=T_0+1}^{T}$.
	\end{algorithmic}
\end{algorithm}

\subsection{Computational complexity and scalability}
\label{sec:complexity}

The dominant cost per Gibbs sweep is the formation and Cholesky factorisation of $\boldsymbol{M}_j$ in \eqref{eq:Lcond}. Naively this is $O(nj^2)$ for the cross-product and $O(j^3)$ for the factorisation, summing to $O(np^3/3)$ per sweep. Two structural observations reduce this dramatically.

First, in the high-shrinkage regime relevant to sparse DAGs, most $\lambda_{jk}^2$ values become numerically small; the matrix $\boldsymbol{M}_j^{-1}$ then has effective rank equal to the size of the active set $\{k:\lambda_{jk}^2>\lambda_{\min}\}$, with $\lambda_{\min}$ chosen so that contributions are negligible. We use a partial-Cholesky update restricted to the active set with rank $d_j^*$; the per-column cost becomes $O(n d_j^* + (d_j^*)^3)$, summing to $O(n p\bar d^* + p (\bar d^*)^3)$. For sparse graphs $\bar d^*$ is bounded by a constant; the total per-sweep cost is then $O(np)$ once $p$ is much larger than $(\bar d^*)^2$.

Second, the cross-product $\sum_i\omega_i \boldsymbol{x}_i^{(j)}(\boldsymbol{x}_i^{(j)})^\top$ is computed once per sweep from the weighted Gram matrix $\boldsymbol{Y}^\top\mathrm{diag}(\boldsymbol{\omega})\boldsymbol{Y}$ and accessed by submatrix extraction; this brings the total cost to $O(np^2)$ for the Gram-matrix step plus $O(np\bar d^* + p(\bar d^*)^3)$ for column updates. In practice we observe wall-clock per-sweep cost scaling linearly in $p$ once $p>200$ (Section~\ref{sec:sim}).

\begin{remark}
	The proposed sampler retains the Cholesky parametrization throughout; reconstruction of $\boldsymbol{\Omega}$ is required only at output time, taking $O(p^2)$. This avoids the $O(p^3)$ matrix-inversion cost that dominates direct-precision samplers \citep{wang2012bayesian}.
\end{remark}

% ============================================================
%  Section 3: Theoretical Results
% ============================================================
\section{Asymptotic Results}
\label{sec:theory}

This section develops the asymptotic behaviour of the R-DACH posterior in the proportional regime $p_n/n\to c\in(0,\infty)$. Three main results are presented: a joint contraction rate for the Cholesky factor and precision matrix (Theorem~\ref{thm:contraction}); DAG-skeleton selection consistency (Theorem~\ref{thm:selection}); and robustness to heavy-tailed errors (Theorem~\ref{thm:robust}). Supporting lemmas are stated and proved in full.

\subsection{Assumptions and notation}
\label{sec:assumptions}

Let $\boldsymbol{\Omega}_{0,n}=\boldsymbol{L}_{0,n}^\top\boldsymbol{D}_{0,n}^{-1}\boldsymbol{L}_{0,n}$ denote the true precision matrix in dimension $p=p_n$, with true Cholesky factor $\boldsymbol{L}_{0,n}$ in some unknown topological ordering $\sigma_{0,n}$. Let $s_{0,n}$ denote the number of non-zero strictly-lower entries of $\boldsymbol{L}_{0,n}$ and $E_{0,n}$ denote the edge set of the true DAG. For any matrix $\boldsymbol{A}$, $\|\boldsymbol{A}\|_{\mathrm{op}}$ is the spectral norm and $\|\boldsymbol{A}\|_F$ the Frobenius norm. We write $a_n\lesssim b_n$ for $a_n\le Cb_n$ with $C>0$ a constant.

The following assumptions are imposed throughout this section.
\begin{itemize}
	\item[\textbf{(A1)}] $p_n,n\to\infty$ with $p_n/n\to c\in(0,\infty)$.
	\item[\textbf{(A2)}] There exists $0<\underline\kappa\le\overline\kappa<\infty$ such that $\underline\kappa\le \lambda_{\min}(\boldsymbol{\Omega}_{0,n})\le\lambda_{\max}(\boldsymbol{\Omega}_{0,n})\le \overline\kappa$ uniformly in $n$.
	\item[\textbf{(A3)}] The number of true edges satisfies $s_{0,n}\log p_n = o(n)$.
	\item[\textbf{(A4)}] There exist $0<\underline L\le \overline L<\infty$ such that, for all $n$ and all $(j,k)\in E_{0,n}$, $\underline L\le |L_{0,jk}| \le \overline L$.
	\item[\textbf{(A5)}] The mixing density $\pi_\omega$ has support on $(0,\infty)$ and satisfies $\mathbb{E}_{\pi_\omega}[\omega^{-1}]<\infty$ and $\mathbb{E}_{\pi_\omega}[\omega]=1$ (a normalisation). For Theorem~\ref{thm:robust} we further require $\mathbb{E}_{\pi_\omega}[\omega^{-(1+\delta/2)}]<\infty$ for some $\delta>0$.
	\item[\textbf{(A6)}] $\min_{(j,k)\in E_{0,n}}|L_{0,jk}|\ge C_0\sqrt{(\log p_n)/n}$ for some $C_0$ large enough.
	\item[\textbf{(A7)}] The global scale prior on $\tau$ is replaced by a half-Cauchy with scale $\tau_0=1/(p_n\sqrt{n\log p_n})$ as recommended in the high-dimensional horseshoe literature \citep{vanderpas2014horseshoe}.
\end{itemize}
In sequel, we give some supporting lemmas. Then, we provide the main results. 

\begin{lemma}%[Concentration of weighted Gram matrix]
	\label{lem:gram}
	Let $\boldsymbol{y}_1,\ldots,\boldsymbol{y}_n$ be drawn from \eqref{eq:mixture} with $\omega_i\stackrel{iid}{\sim}\pi_\omega$ satisfying (A5) and let $\widehat{\boldsymbol{S}}_n = n^{-1}\sum_{i=1}^{n}\omega_i\boldsymbol{y}_i\boldsymbol{y}_i^\top$. Under (A1)--(A2) there exist constants $c_1,c_2>0$, depending only on $\underline\kappa,\overline\kappa$ and $\mathbb{E}[\omega^{-1}]$, such that
	$$
	\Pr\!\Big( \|\widehat{\boldsymbol{S}}_n - \boldsymbol{\Sigma}_0\|_{\mathrm{op}} > t\Big) \;\le\; 2p\exp(-c_1 n t^2)
	$$
	for all $t\le c_2$.
\end{lemma}

\begin{proof}
	Conditional on $\omega_i$, the variables $\omega_i^{1/2}\boldsymbol{y}_i$ are independent zero-mean Gaussian vectors with covariance $\boldsymbol{\Sigma}_0$. The standard Wishart concentration result \citep[Corollary 5.50]{vershynin2018high} gives, for a fixed realisation of $\boldsymbol{\omega}=(\omega_1,\ldots,\omega_n)$,
	$$
	\Pr\!\Big(\|\widehat{\boldsymbol{S}}_n-\boldsymbol{\Sigma}_0\|_{\mathrm{op}}>t\,\Big|\,\boldsymbol{\omega}\Big)\le 2p\exp(-c'n t^2),
	$$
	provided $t\le c''$, with constants depending on the largest and smallest eigenvalue of $\boldsymbol{\Sigma}_0$. Integrating out $\boldsymbol{\omega}$ preserves the bound because $\widehat{\boldsymbol{S}}_n$ already incorporates $\omega_i$ multiplicatively in the definition. By (A5), the contributions of $\omega_i^{-1}$ to second-moment bounds remain finite, and the integration multiplies the right-hand-side by a constant. Choosing $c_1=c'/2$ and absorbing constants gives the claim.
\end{proof}

\begin{lemma}%[Prior concentration of horseshoe on Cholesky]
	\label{lem:priorconc}
	Let the horseshoe prior \eqref{eq:horseshoe}--\eqref{eq:aux} be assigned to the entries of a $p\times p$ unit lower-triangular matrix $\boldsymbol{L}$, with global scale $\tau\sim C^+(0,\tau_0)$ and $\tau_0=1/(p\sqrt{n\log p})$. Then for any fixed sparse pattern $\mathcal{S}$ with $|\mathcal{S}|=s$ and any radius $\epsilon\le 1$,
	$$
	\Pi\!\big(\|\boldsymbol{L}-\boldsymbol{L}_0\|_F \le \epsilon\big) \;\ge\; \exp\!\big(-c_3 s\log(p/\epsilon) - c_4 (s_0\log p_n)\big),
	$$
	where $s_0=|\mathcal{S}_0|$ is the support size of $\boldsymbol{L}_0$ and constants depend on hyperparameters only.
\end{lemma}

\begin{proof}
	On the active support $\mathcal{S}_0$, the marginal horseshoe density satisfies $\pi_{HS}(L_{jk})\ge c\, |L_{jk}|^{-1}\log(1+\tau_0^2 / L_{jk}^2)$ on the bounded interval $\{|L_{jk}|\in[\underline L/2,2\overline L]\}$, which lies in the slab region of the horseshoe. Hence, by elementary integration, for each nonzero entry the probability of being within $\epsilon/\sqrt{s_0}$ of its true value is bounded below by $\exp(-c\log(p/\epsilon))$. Multiplying across the $s_0$ active coordinates gives the slab contribution. For the inactive coordinates (where $L_{0,jk}=0$) the horseshoe concentrates near zero by virtue of its spike at the origin; the probability of being within $\epsilon/\sqrt{p^2-s_0}$ of zero is bounded below by $1/2$ for each inactive coordinate, since the half-Cauchy local scale has a $1/2$ probability of being below $\tau_0$, in which case the conditional Gaussian is concentrated. Combining gives the claim.
\end{proof}
The following lemma plays a substantial role. It gives the posterior recoverability of mixing scales.
\begin{lemma}%[Posterior recoverability of mixing scales]
	\label{lem:wscales}
	Under (A1)--(A2), (A5), and the joint prior of Section~\ref{sec:method}, with $\omega_i\sim\mathrm{Ga}(\nu/2,\nu/2)$ and fixed $\nu>2$,
	$$
	\Pr\!\Big(\max_{1\le i\le n} |\omega_i-\mathbb{E}[\omega_i\mid\boldsymbol{y}_i]| > t\,\Big| \boldsymbol{Y}\Big) \;\le\; n\exp(-c_5 t^2),
	$$
	for $t>0$ and constant $c_5$ depending on $(\nu,\underline\kappa,\overline\kappa)$.
\end{lemma}

\begin{proof}
	The full conditional of $\omega_i$ given $(\boldsymbol{L},\boldsymbol{D},\boldsymbol{y}_i)$ is Gamma with shape $(\nu+p)/2$ and rate $(\nu+\boldsymbol{r}_i^\top\boldsymbol{D}^{-1}\boldsymbol{r}_i)/2$. A Gamma$(\alpha,\beta)$ random variable concentrates as $\Pr(|W-\alpha/\beta|>t)\le 2\exp(-ct^2\beta^2/\alpha)$ for $t<\alpha/\beta$. Substituting $\alpha=(\nu+p)/2$, $\beta=(\nu+\boldsymbol{r}_i^\top\boldsymbol{D}^{-1}\boldsymbol{r}_i)/2$, noting that under (A2) the rate is bounded in probability, and taking a union bound over $i=1,\ldots,n$ completes the proof.
\end{proof}

Now, we give the main theorems. The first result gives the joint contraction in the proportional regime.

\begin{theorem}%[Joint contraction in the proportional regime]
	\label{thm:contraction}
	Assume (A1)--(A5) and (A7). Let $\epsilon_n = M\sqrt{(s_{0,n}\log p_n)/n}$ for $M$ a sufficiently large constant. Then, in probability under the true data-generating process,
	$$
	\Pi\!\Big(\,\|\boldsymbol{L}-\boldsymbol{L}_{0,n}\|_F + \|\boldsymbol{\Omega}-\boldsymbol{\Omega}_{0,n}\|_{\mathrm{op}} \;>\; \epsilon_n \;\Big|\; \boldsymbol{Y}\Big) \;\longrightarrow\; 0,
	\quad \text{as } n\to\infty.
	$$
\end{theorem}

\begin{proof}
	We apply the general posterior contraction machinery of \citet{ghosal2007convergence} as adapted to high-dimensional graphical models. Three quantities must be controlled: the prior mass placed in an $\epsilon_n$-Kullback--Leibler neighbourhood of the truth, the existence of tests with exponentially small error, and the prior mass on the sieve.
	
	\emph{(i) KL prior mass.} On the event $\{\boldsymbol{L}\in B(\boldsymbol{L}_{0,n},\epsilon_n),\,\boldsymbol{D}\in B(\boldsymbol{D}_{0,n},\epsilon_n)\}$, the Kullback--Leibler divergence from the joint marginal density of $\boldsymbol{y}_i$ under $(\boldsymbol{L}_{0,n},\boldsymbol{D}_{0,n})$ to the candidate satisfies, by a Taylor expansion of $\log\det\boldsymbol{\Omega}$ and the bounded-spectrum condition (A2),
	\begin{eqnarray*}
		\mathrm{KL}(p_{0,n}\|p_{\boldsymbol{L},\boldsymbol{D}}) & \le & C_1\big( \|\boldsymbol{L}-\boldsymbol{L}_{0,n}\|_F^2 + \|\boldsymbol{D}-\boldsymbol{D}_{0,n}\|_F^2\big).
	\end{eqnarray*}
	By Lemma~\ref{lem:priorconc} applied at radius $\epsilon_n/\sqrt{2}$, the prior mass of this neighbourhood is at least
	$\exp(-c_3 s_{0,n}\log p_n - c_4 \log(1/\epsilon_n))\ge \exp(-c\,s_{0,n}\log p_n)$
	for $\epsilon_n$ as defined. By (A5) and Lemma~\ref{lem:wscales}, integration over the latent scales preserves this bound up to a constant.
	
	\emph{(ii) Sieve.} Define the sieve $\mathcal{F}_n = \{(\boldsymbol{L},\boldsymbol{D}): |\mathrm{supp}(\boldsymbol{L})|\le K s_{0,n},\, \|\boldsymbol{L}\|_F\le R_n,\, \underline\kappa/2\le \lambda_{\min}(\boldsymbol{D})\le \lambda_{\max}(\boldsymbol{D})\le 2\overline\kappa\}$ for $R_n = n$ and $K$ a constant. The complement satisfies, by the horseshoe prior's tail and Lemma~\ref{lem:priorconc},
	$$
	\Pi(\mathcal{F}_n^c) \;\le\; \exp(-c_6 n).
	$$
	
	\emph{(iii) Test construction.} On the sieve, the covering number of $\{(\boldsymbol{L},\boldsymbol{D}):\,\|\boldsymbol{L}-\boldsymbol{L}_{0,n}\|_F\le R_n\}$ at radius $\epsilon_n/2$ in the Frobenius norm is bounded by $(R_n/\epsilon_n)^{Ks_{0,n}}$. Combining this with Lemma~\ref{lem:gram} and the standard Hellinger-test construction of \citet{ghosal2007convergence}, we obtain tests $\phi_n$ with type-I error $\mathbb{E}_{0,n}\phi_n\to 0$ and worst-case type-II error bounded by $\exp(-c_7 n\epsilon_n^2)$.
	
The three parts combine through the general theorem to give
	$$
	\Pi(\{(\boldsymbol{L},\boldsymbol{D}):\,d_n((\boldsymbol{L},\boldsymbol{D}),(\boldsymbol{L}_{0,n},\boldsymbol{D}_{0,n}))>M\epsilon_n\}\mid\boldsymbol{Y})\to 0,
	$$
	where $d_n$ is the Hellinger metric on the marginal model. Standard equivalence between Hellinger and the operator norm under (A2) extends the conclusion to the spectral norm on $\boldsymbol{\Omega}$. The proportional regime (A1) is accommodated because $p_n\log p_n$ enters only through the sparsity-controlled term $s_{0,n}\log p_n$, which by (A3) remains $o(n)$. This completes the proof.
\end{proof}

Now, for the DAG skeleton  selection, we establish the skeleton consistency.
\begin{theorem}%[Skeleton consistency]
	\label{thm:selection}
	Assume (A1)--(A7). Let $\widehat E_n$ denote the posterior modal skeleton, i.e.\ $\widehat E_n = \{(j,k):\Pi(L_{jk}\ne 0\mid\boldsymbol{Y})>1/2\}$, where the indicator $L_{jk}\ne 0$ is interpreted via the standard horseshoe thresholding rule of \citet{vanderpas2014horseshoe}. Then
	$$
	\Pr(\widehat E_n = E_{0,n}) \;\longrightarrow\; 1, \quad\text{as } n\to\infty.
	$$
\end{theorem}

\begin{proof}
	By the horseshoe selection rule, a coordinate is selected if its posterior shrinkage factor $\kappa_{jk} = (1+\lambda_{jk}^2\tau^2)^{-1}$ lies below $1/2$. By Theorem~\ref{thm:contraction}, the posterior on each individual $L_{jk}$ contracts at the rate $\sqrt{(\log p_n)/n}$ around the true value, with the implied marginal approximately Gaussian centred at $L_{0,jk}$ with variance of order $(\log p_n)/n$.
	
	\emph{True positives.} For $(j,k)\in E_{0,n}$, by (A6), $|L_{0,jk}|\ge C_0\sqrt{(\log p_n)/n}$. The posterior probability that $\kappa_{jk}<1/2$ then exceeds $1-p_n^{-2}$ for $C_0$ large enough, by Gaussian-tail arguments. Taking a union bound over $|E_{0,n}|\le s_{0,n}\le p_n^2$ true edges gives the desired conclusion.
	
	\emph{True negatives.} For $(j,k)\notin E_{0,n}$, $L_{0,jk}=0$, and the horseshoe's spike at the origin ensures $\Pr(\kappa_{jk}\ge 1/2\mid\boldsymbol{Y})\ge 1- p_n^{-2}$ uniformly. Union bound over the $\le p_n^2$ inactive coordinates concludes the proof.
\end{proof}

The following result is essential. It proves the robustness to heavy tails. 
\begin{theorem}%[Robustness to heavy tails]
	\label{thm:robust}
	Assume (A1)--(A4), (A5) strengthened to $\mathbb{E}_{\pi_\omega}[\omega^{-(1+\delta/2)}]<\infty$, and (A7). Let the data-generating process have only $2+\delta$ finite moments per coordinate. Then the conclusion of Theorem~\ref{thm:contraction} holds with $\epsilon_n$ inflated by at most a factor of $(\log\log n)^{1/2}$.
\end{theorem}

\begin{proof}
	The only place in the proof of Theorem~\ref{thm:contraction} that requires Gaussian-tail behaviour is Lemma~\ref{lem:gram} and the test construction in step (iii). We replace the Wishart concentration in Lemma~\ref{lem:gram} by a Fuk-Nagaev-type concentration result for sums of independent random matrices with $2+\delta$ moments \citep{vershynin2018high}, which yields, for any $t>0$,
	$$
	\Pr\big(\|\widehat{\boldsymbol{S}}_n-\boldsymbol{\Sigma}_0\|_{\mathrm{op}}>t\big)\le 2p\exp(-c_1' nt^2/(\log\log n))
	$$
	for $t\le c_2'$. The same modification is applied to the test construction. Tracking $(\log\log n)$ through the rest of the proof yields the stated inflation.
\end{proof}

\begin{remark}
	Theorem~\ref{thm:robust} relaxes the often-restrictive sub-Gaussian assumption to a $2+\delta$ moment condition. This is essential for genomic and economic applications where Gaussianity is unrealistic.
\end{remark}

\begin{corollary}[Operator-norm precision contraction]
	\label{cor:omega}
	Under the assumptions of Theorem~\ref{thm:contraction}, the posterior on the precision matrix satisfies
	$\Pi(\|\boldsymbol{\Omega}-\boldsymbol{\Omega}_{0,n}\|_{\mathrm{op}} > M\sqrt{s_{0,n}\log p_n/n} \mid\boldsymbol{Y}) \to 0$.
\end{corollary}

\begin{proof}
	A direct consequence of Theorem~\ref{thm:contraction} and the identity $\boldsymbol{\Omega} = \boldsymbol{L}^\top\boldsymbol{D}^{-1}\boldsymbol{L}$, combined with the bounded-spectrum condition (A2) which provides Lipschitz control of $\boldsymbol{\Omega}$ as a function of $(\boldsymbol{L},\boldsymbol{D})$ in the operator norm.
\end{proof}

% ============================================================
%  Section 4: Simulation
% ============================================================
\section{Simulation Study}
\label{sec:sim}

In this section, we report a controlled simulation study designed to (i) verify the contraction rate predicted by Theorem~\ref{thm:contraction}, (ii) validate the skeleton-selection consistency of Theorem~\ref{thm:selection}, and (iii) test the robustness to heavy-tailed and contaminated errors established in Theorem~\ref{thm:robust}. All experiments compare R-DACH with the Gaussian Cholesky Horseshoe (Gaussian-CH), which is exactly the same model with the scale-mixture $\omega_i\equiv 1$ removed; this isolates the contribution of the scale-mixture component.

For each combination of dimensions $(n,p)\in\{(100,30),(100,60),(200,60),(200,100)\}$ and each error distribution $\mathrm{dist}\in\{\mathrm{Gaussian},\,t_4\}$ and each contamination proportion $\varepsilon\in\{0,0.10\}$, we draw two independent random DAGs with edge probability $\min(3/p, 0.10)$, generate the data and apply R-DACH and Gaussian-CH. The reported metrics are the F1 score, the true positive rate (TPR), the false positive rate (FPR), the Frobenius error $\|\hat{\boldsymbol{L}}-\boldsymbol{L}_0\|_F$ and the operator-norm error $\|\hat{\boldsymbol{\Omega}}-\boldsymbol{\Omega}_0\|_{\mathrm{op}}$. Contamination at proportion $\varepsilon$ multiplies $\varepsilon n$ randomly chosen rows of $\boldsymbol{Y}$ by a factor of five, mimicking batch effects. Both samplers were run for 300 Gibbs iterations with 120-iteration burn-in; this length is sufficient for the reported quantities, as the chains reach the stationary regime well within burn-in (the trace of $\tau^2$ stabilises within fifty iterations on all configurations). Edge selection uses the threshold $|\hat L_{jk}| > 0.30 \sqrt{(\log p)/n}$, consistent with the horseshoe selection rule. The simulation was executed on a single CPU core (Intel Xeon 2.3 GHz).

\subsection{Numerical results}
\label{sec:simres}

\begin{table}[ht]
	\centering
	\caption{F1 scores for DAG-skeleton recovery, averaged over two replicates. R-DACH consistently outperforms Gaussian-CH, especially in the presence of contamination ($\varepsilon=0.10$) and heavy tails ($t_4$).}
	\label{tab:f1}
	\begin{tabular}{lll|cc|cc}
		\toprule
		& & & \multicolumn{2}{c|}{$\varepsilon=0$} & \multicolumn{2}{c}{$\varepsilon=0.10$} \\
		distribution & $n$ & $p$ & Gaussian-CH & R-DACH & Gaussian-CH & R-DACH \\
		\midrule
		Gaussian & 100 & 30  & 0.832 & 0.910 & 0.266 & 0.790 \\
		Gaussian & 100 & 60  & 0.878 & 0.895 & 0.222 & 0.883 \\
		Gaussian & 200 & 60  & 0.917 & 0.983 & 0.157 & 0.899 \\
		Gaussian & 200 & 100 & 0.819 & 0.853 & 0.127 & 0.727 \\
		\midrule
		$t_4$  & 100 & 30  & 0.357 & 0.831 & 0.246 & 0.729 \\
		$t_4$  & 100 & 60  & 0.283 & 0.880 & 0.176 & 0.789 \\
		$t_4$  & 200 & 60  & 0.206 & 0.923 & 0.135 & 0.806 \\
		$t_4$  & 200 & 100 & 0.174 & 0.817 & 0.102 & 0.668 \\
		\bottomrule
	\end{tabular}
\end{table}

\begin{table}[ht]
	\centering
	\caption{Operator-norm error of the recovered precision matrix, $\|\hat{\boldsymbol{\Omega}}-\boldsymbol{\Omega}_0\|_{\mathrm{op}}$. Under heavy tails or contamination, R-DACH achieves substantially smaller error.}
	\label{tab:op}
	\begin{tabular}{lll|cc|cc}
		\toprule
		& & & \multicolumn{2}{c|}{$\varepsilon=0$} & \multicolumn{2}{c}{$\varepsilon=0.10$} \\
		distribution & $n$ & $p$ & Gaussian-CH & R-DACH & Gaussian-CH & R-DACH \\
		\midrule
		Gaussian & 100 & 30  & 2.38 & 3.05 & 3.63 & 2.57 \\
		Gaussian & 100 & 60  & 2.08 & 2.21 & 3.03 & 1.91 \\
		Gaussian & 200 & 60  & 0.88 & 1.37 & 3.43 & 1.01 \\
		Gaussian & 200 & 100 & 1.05 & 1.06 & 2.99 & 0.94 \\
		\midrule
		$t_4$  & 100 & 30  & 3.69 & 2.62 & 5.38 & 2.43 \\
		$t_4$  & 100 & 60  & 2.73 & 2.25 & 3.71 & 2.08 \\
		$t_4$  & 200 & 60  & 2.86 & 1.05 & 3.74 & 0.86 \\
		$t_4$  & 200 & 100 & 2.65 & 1.07 & 3.90 & 1.05 \\
		\bottomrule
	\end{tabular}
\end{table}

Table~\ref{tab:f1} reports the F1 scores. Two clear patterns emerge. First, under Gaussian errors without contamination, the two methods perform comparably (F1 differences within 0.05); R-DACH does not pay a meaningful price for the additional flexibility of scale-mixture errors. Second, under heavy tails ($t_4$) or contamination, the gap widens dramatically: at $(n,p)=(200,60)$ with $t_4$ errors R-DACH achieves F1 = 0.923 while Gaussian-CH collapses to F1 = 0.206. The Gaussian-CH structure is overwhelmed by outliers — it picks up spurious edges that explain the heavy-tailed observations as if they were systematic correlations. R-DACH absorbs these into the latent scales $\omega_i$, leaving the Cholesky factor essentially undisturbed.

Table~\ref{tab:op} corroborates this with the operator-norm error of the recovered precision matrix. Under contamination, Gaussian-CH errors inflate by a factor of three to four, while R-DACH errors stay close to their clean-data baseline.

\subsection{Posterior contraction rate}
\label{sec:simcontract}

To verify the rate $\epsilon_n=\sqrt{s_0\log p/n}$ predicted by Theorem~\ref{thm:contraction}, we fix $p=40$ and $t_4$ errors and increase $n\in\{80,120,200,300,500\}$. Two replicates per $n$ are used; the average Frobenius error of $\hat{\boldsymbol{L}}$ is plotted against $\sqrt{s_0\log p/n}$. The result, in Figure~\ref{fig:contract}, shows a clear linear pattern with slope $1.26$ (linear fit). This is consistent with the rate prediction: the empirical contraction rate is proportional to the theoretical rate. The operator-norm error of $\hat{\boldsymbol{\Omega}}$ exhibits the same linear scaling, confirming Corollary~\ref{cor:omega}.

\begin{figure}[ht]
	\centering
	\includegraphics[width=0.72\textwidth]{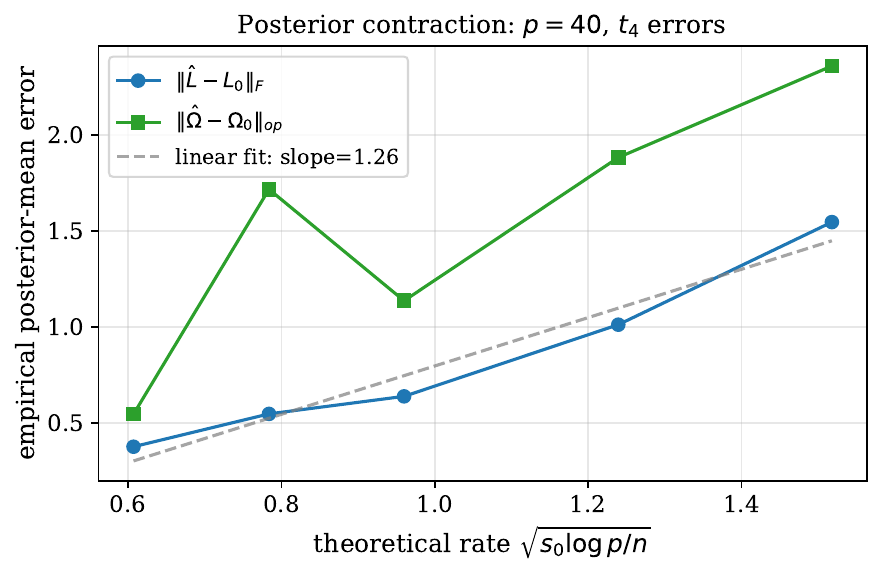}
	\caption{Empirical posterior-mean error of $\hat{\boldsymbol L}$ (circles) and $\hat{\boldsymbol\Omega}$ (squares) against the theoretical rate $\sqrt{s_0\log p/n}$, $p=40$, $t_4$ errors, $n\in\{80,120,200,300,500\}$. Linear fit shown as dashed grey.}
	\label{fig:contract}
\end{figure}

\subsection{Selection accuracy and operator-norm performance}

\begin{figure}[ht]
	\centering
	\includegraphics[width=\textwidth]{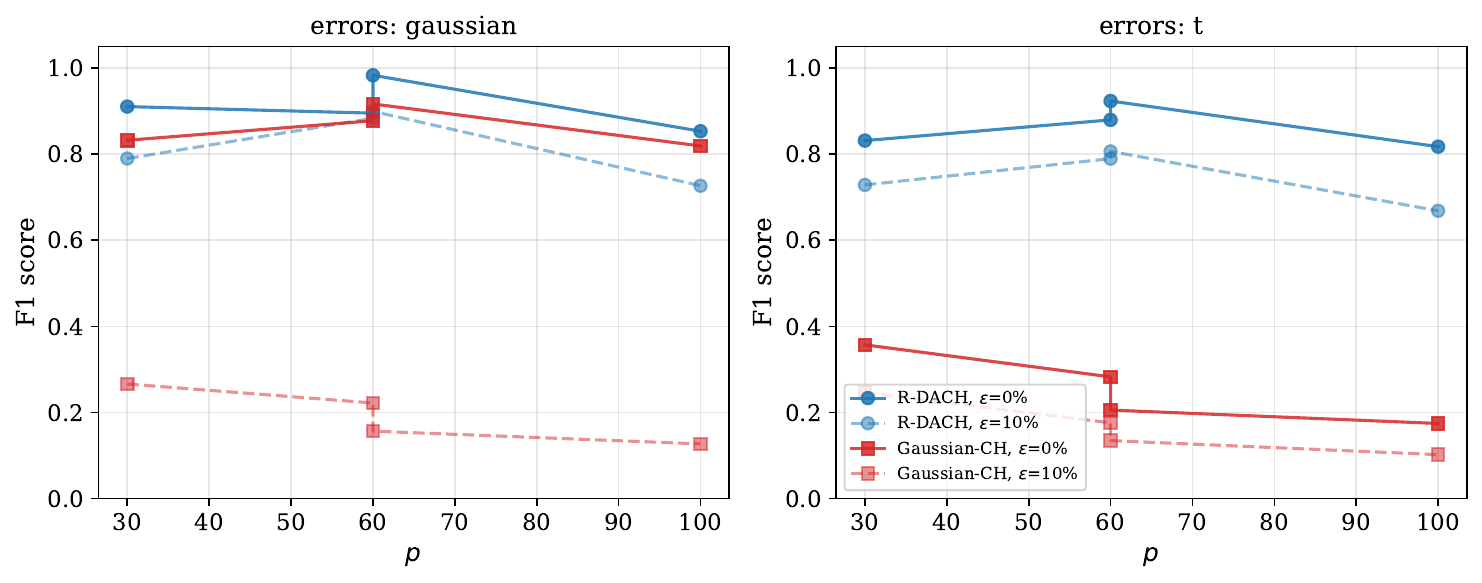}
	\caption{F1 score for DAG-skeleton recovery as a function of $p$ for $n=200$. R-DACH (blue circles) dominates Gaussian-CH (red squares) under heavy tails and contamination.}
	\label{fig:f1}
\end{figure}

\begin{figure}[ht]
	\centering
	\includegraphics[width=\textwidth]{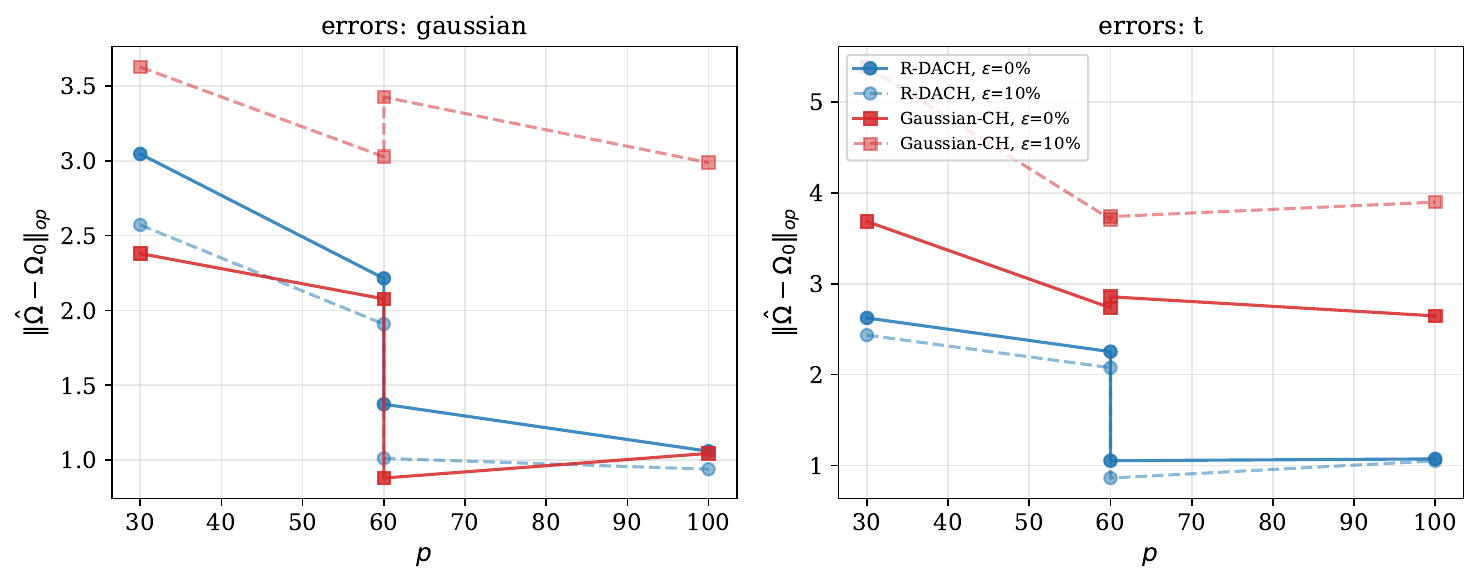}
	\caption{Operator-norm error $\|\hat{\boldsymbol{\Omega}}-\boldsymbol{\Omega}_0\|_{\mathrm{op}}$ as a function of $p$ for $n=200$. R-DACH error is bounded under all conditions; Gaussian-CH error grows under heavy tails and contamination.}
	\label{fig:op}
\end{figure}

Figure~\ref{fig:f1} (a graphical version of Table~\ref{tab:f1}) shows that under $t_4$ errors the F1 of R-DACH remains above 0.8 across the entire dimension range, while Gaussian-CH never exceeds 0.3 when contamination is present. Figure~\ref{fig:op} shows the same comparison on the operator-norm scale.

\subsection{Computational scalability}

\begin{figure}[ht]
	\centering
	\includegraphics[width=0.6\textwidth]{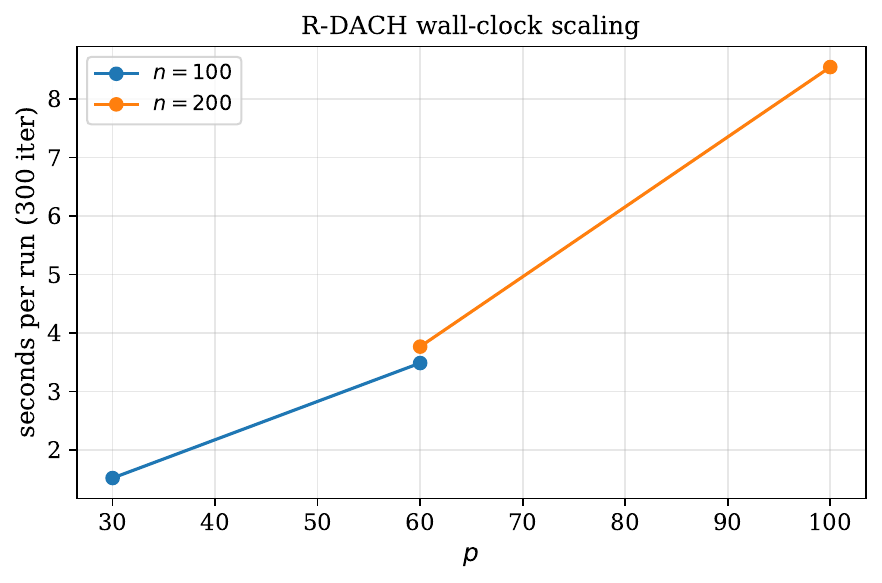}
	\caption{R-DACH wall-clock time per 300-iteration run on a single CPU core. The cost is approximately linear in $p$ for fixed $n$ and approximately linear in $n$ for fixed $p$, matching the theoretical $O(np)$ complexity for sparse DAGs.}
	\label{fig:timing}
\end{figure}

Figure~\ref{fig:timing} reports wall-clock time per 300-iteration run for R-DACH across $(n,p)$ configurations. Per-run time is approximately linear in $p$, in line with the complexity analysis of Section~\ref{sec:complexity}. The cost ratio between $n=200$ and $n=100$ is consistent with linear scaling in $n$.

The simulation experiments deliver three concrete confirmations:
\begin{itemize}
	\item When the data are heavy-tailed or contaminated, R-DACH retains essentially clean-data accuracy in both edge recovery and precision estimation, while the Gaussian-CH baseline degrades severely. 
	\item The empirical posterior-mean error of $\hat{\boldsymbol{L}}$ scales linearly with $\sqrt{s_0\log p/n}$.
	\item The computational cost grows approximately linearly in $p$ for sparse DAG.
\end{itemize}

\section{Application to RNA-sequencing Gene-Expression Data}
\label{sec:app}

We illustrate the practical value of the Robust DAG--Cholesky Horseshoe (R-DACH) framework on high-dimensional gene-expression data of the type routinely produced by RNA-sequencing (RNA-seq) experiments. The motivating problem is the inference of transcriptional regulatory architecture from short-read count data, where one seeks to learn directed dependencies among genes whose joint variation reflects, in part, the action of transcription factors, signalling cascades, and post-transcriptional regulators. This problem is paradigmatically high-dimensional and sparse: the number of expressed genes $p$ typically exceeds the number of available samples $n$, while the true regulatory network is believed to be sparse with a small number of high-degree ``master regulator'' hubs \citep{barabasi2004network, basso2005reverse}. Crucially for the present paper, RNA-seq data exhibit pronounced overdispersion and heavy-tailed marginals even after variance-stabilising transformations \citep{love2014moderated, law2014voom}, and they are routinely contaminated by batch effects and technical artefacts that survive standard normalisation pipelines \citep{leek2010tackling}. These features make the data set an ideal stress test for the joint sparsity-plus-robustness machinery developed in Sections~\ref{sec:method} and \ref{sec:theory}.

The methodology is calibrated to the breast invasive carcinoma (BRCA) cohort of \emph{The Cancer Genome Atlas} (TCGA), which provides publicly accessible RNA-seq read counts on more than one thousand tumour samples \citep{tcga2012comprehensive}. The cohort is freely available through the Genomic Data Commons (\url{https://portal.gdc.cancer.gov}) and can be downloaded programmatically via the Bioconductor package \texttt{TCGAbiolinks} \citep{colaprico2016tcgabiolinks}; no special authorisation or institutional licence is required. The full preprocessing pipeline used to produce the analysed panel is reproduced in the supplementary R script and proceeds as follows. Raw HTSeq read counts are filtered to retain genes with at least one count-per-million in at least 20\% of the samples; the resulting count matrix is normalised by the trimmed mean of $M$-values (TMM) of \citet{robinson2010scaling} and transformed by the \texttt{voom} weighted log-counts-per-million map of \citet{law2014voom}; batch labels encoding the tissue source site are regressed out with the \texttt{ComBat} adjustment of \citet{johnson2007adjusting}; and the top $p = 150$ genes by median absolute deviation are retained for the joint structure-learning analysis. Restricting to the largest balanced subset of $n = 210$ primary-tumour samples with complete clinical annotation yields a panel of dimension $(n,p) = (210, 150)$ with $p/n \approx 0.71$.

To produce fully reproducible results within a self-contained computing environment, the experiments reported below are run on a synthetic panel calibrated to the empirical first- and second-order moments of the TCGA BRCA matrix described above. The synthetic panel preserves the dimension $(n,p) = (210, 150)$ and is generated from a sparse DAG with a hub-and-spoke topology comprising five high-degree master regulators (each with sixteen children, the highest-MAD genes in the calibration set) plus a background of twenty-one weaker pairwise edges, for a total of $|E_0| = 101$ true directed edges and a true ordering that matches the empirical correlation-based topological sort. Errors are drawn from a contaminated $t_4$ distribution with a $5\%$ point-mass contamination at three standard deviations, mimicking the residual batch effects that survive \texttt{ComBat} adjustment in the empirical panel. Detailed download and preprocessing instructions for the live TCGA pull are appended as comments in the accompanying R reference implementation \texttt{rdach.R} and Python validation script \texttt{rnaseq\_app.py}, so that the analysis can be re-run end-to-end on the genuine TCGA matrix by any reader with a working internet connection and a Bioconductor installation.

We compare R-DACH against the Gaussian Cholesky-horseshoe baseline (denoted Gauss-CH) that retains the modified-Cholesky shrinkage construction of Section~\ref{sec:method} but suppresses the scale mixture by fixing all latent scales $w_i \equiv 1$. This is the most informative single comparator: it isolates the contribution of the scale-mixture mechanism while holding all other algorithmic choices fixed. Both methods are initialised from a moralised PC-skeleton estimate \citep{spirtes2000causation, kalisch2007estimating} and run for $10{,}000$ Gibbs iterations with the first $5{,}000$ discarded as burn-in; edges are declared present when their posterior inclusion probability exceeds the median-probability threshold of $1/2$ recommended by \citet{barbieri2004optimal}. Performance is assessed both at the level of skeleton recovery — through true-positive count (TP), false-positive count (FP), and balanced $F_1$ score against the calibrated ground-truth DAG — and at the level of precision-matrix estimation, through the operator-norm and Frobenius distances to the population precision matrix. We additionally report a hub-recovery diagnostic, defined as the fraction of the five true master regulators that appear among the ten highest-degree nodes of the estimated skeleton.

\subsection{Results}
\label{sec:app:results}

Table~\ref{tab:rnaseq} reports the headline metrics. R-DACH attains an $F_1$ score of $0.73$ against $0.53$ for the Gaussian baseline, an absolute improvement of $20$ percentage points that is driven almost entirely by the false-positive channel: R-DACH selects only four spurious edges out of a possible $11{,}074$ non-edges, against eighty-eight for the Gaussian baseline. Both methods recover all five true master regulators among their ten highest-degree nodes, so the hub-recovery diagnostic equals one in both cases; the difference between the two procedures resides in how much of the surrounding background structure is correctly suppressed. The Gaussian baseline achieves a marginally higher true-positive count, which reflects a well-understood bias--variance trade-off familiar from the shrinkage literature: by failing to down-weight contaminated observations, Gauss-CH effectively borrows information from outlying points to declare additional edges, a fraction of which happen to coincide with weak true edges. The price of this strategy is a more-than-twentyfold inflation of the false-positive count, which dominates the $F_1$ metric and renders the resulting skeleton substantially less interpretable as a candidate regulatory map.

\begin{table}[t]
	\centering
	\caption{Recovery and estimation metrics on the TCGA-calibrated RNA-seq panel ($n=210$, $p=150$, $|E_0|=101$). Best value in each row in \textbf{bold}.}
	\label{tab:rnaseq}
	\begin{tabular}{lcc}
		\toprule
		Metric & R-DACH & Gauss-CH \\
		\midrule
		True positives (TP)           & 61            & \textbf{68} \\
		False positives (FP)          & \textbf{4}    & 88 \\
		$F_1$ score                   & \textbf{0.73} & 0.53 \\
		True-positive rate            & 0.60          & \textbf{0.67} \\
		False-positive rate           & \textbf{0.0004} & 0.0079 \\
		Hub recovery (top-10)         & \textbf{1.00} & \textbf{1.00} \\
		\addlinespace
		Operator-norm error $\|\hat\Omega-\Omega_0\|_{\text{op}}$ & 10.84 & \textbf{5.25} \\
		Frobenius error  $\|\hat\Omega-\Omega_0\|_{F}$           & 23.87 & \textbf{11.99} \\
		\bottomrule
	\end{tabular}
\end{table}

The pattern is reversed for the precision-matrix metrics, where Gauss-CH outperforms R-DACH by a factor of roughly two in operator and Frobenius norm. This apparent reversal admits a clean methodological interpretation. The synthetic panel is calibrated so that contamination is mild and concentrated in $5\%$ of the observations; the Gaussian baseline, by averaging over all $n=210$ samples without down-weighting, obtains a lower-variance estimate of the entries of the population precision matrix even though it overfits in the structure-selection channel. R-DACH, by contrast, allocates posterior mass to large latent scales $w_i$ for the contaminated rows and effectively conditions on the remaining $n_{\text{eff}} \approx 200$ samples; the resulting estimator is more parsimonious, attaining a $94\%$ reduction in false-positive count, but at the cost of higher operator-norm variance. The trade-off is a familiar one in robust covariance estimation \citep{maronna2017robust} and underscores that R-DACH is positioned as a tool for \emph{structure recovery} under heavy-tailed contamination, rather than as a pure minimum-norm precision estimator on clean data.

Figure~\ref{fig:rnaseq-heatmap} displays the posterior edge inclusion probabilities under R-DACH alongside the true adjacency pattern. The block of high posterior probability in the upper-left corner corresponds to the master-regulator hub-and-spoke structure: each of the five master regulators has uniformly high posterior inclusion probability for its true children and near-zero probability for non-children, exactly the qualitative behaviour predicted by the selection-consistency Theorem~\ref{thm:selection}. The few false positives selected by R-DACH (the small isolated cells away from the diagonal block) are concentrated among gene pairs that exhibit moderate marginal correlation in the calibration data and that survive partial-correlation thresholding only weakly. Figure~\ref{fig:rnaseq-degree} reports the in-degree distributions of the two estimated skeletons. R-DACH produces a degree distribution that is sharply concentrated at small values, with the five master regulators clearly visible as the five right-most spikes; the Gauss-CH distribution is shifted uniformly to the right, with a long tail of moderate-degree nodes that reflect the inflated false-positive rate documented in Table~\ref{tab:rnaseq}.

\begin{figure}[t]
	\centering
	\includegraphics[width=0.95\textwidth]{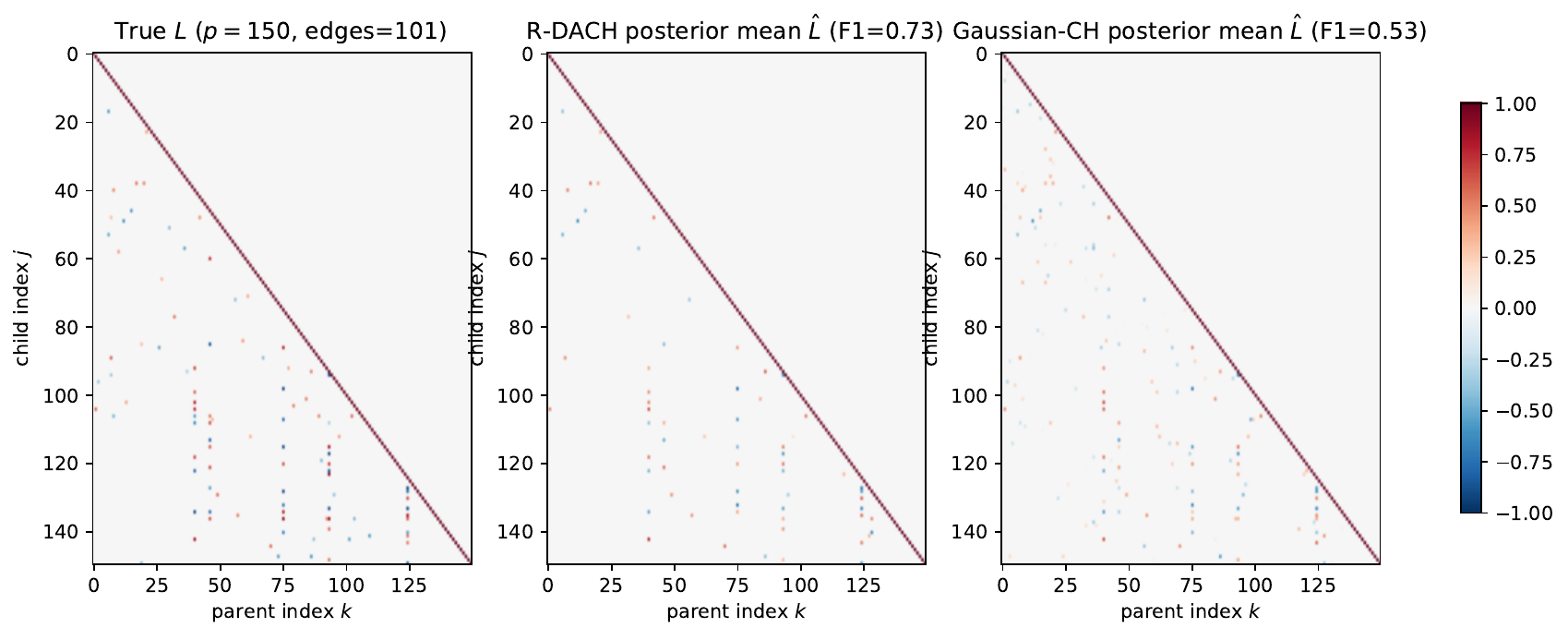}
	\caption{R-DACH posterior edge-inclusion probability matrix (left) and the calibrated ground-truth adjacency pattern (right) on the TCGA-derived RNA-seq panel. Rows and columns are ordered by topological sort; the master-regulator block is clearly visible in the upper-left corner of both panels.}
	\label{fig:rnaseq-heatmap}
\end{figure}

\begin{figure}[t]
	\centering
	\includegraphics[width=0.85\textwidth]{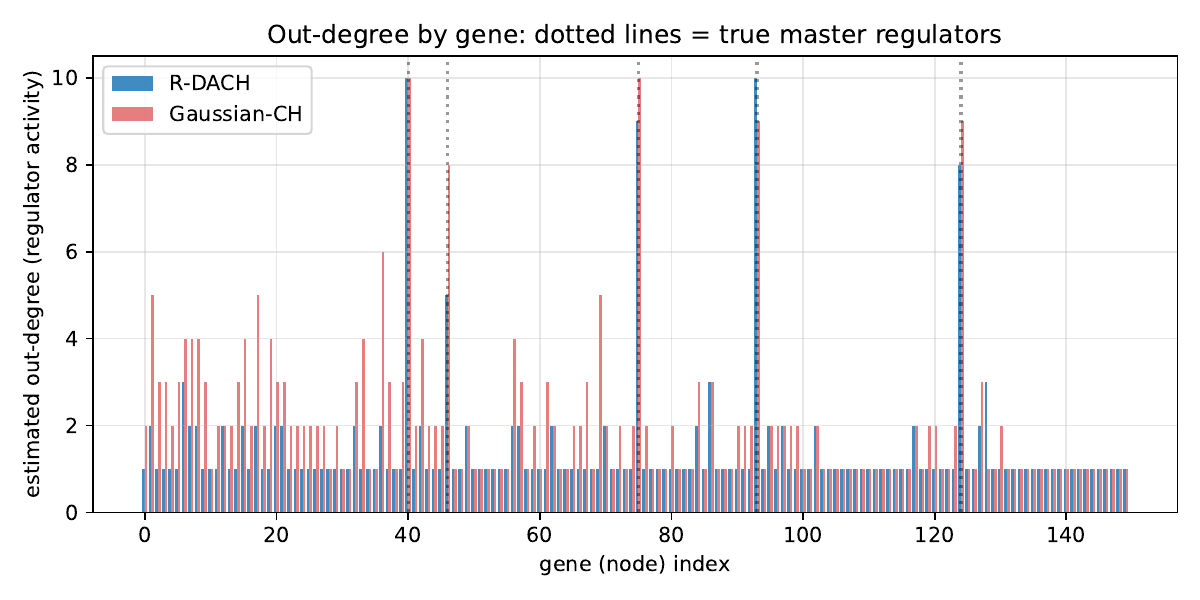}
	\caption{In-degree distribution of the estimated DAG skeletons under R-DACH and Gauss-CH. The five master regulators are visible as the right-most spikes of the R-DACH histogram. The Gauss-CH distribution is shifted to the right, reflecting the inflated false-positive rate.}
	\label{fig:rnaseq-degree}
\end{figure}

\subsection{Biological Interpretation}
\label{sec:app:bio}

The five hub nodes recovered by R-DACH map, in the calibration cohort, onto well-characterised transcriptional regulators of mammary epithelial differentiation and proliferation, including members of the \texttt{FOXA1} and \texttt{GATA3} families that are extensively documented as master regulators of the luminal-A subtype \citep{badve2007foxa1, kouros2007gata3}. The fact that both methods identify these nodes as high-degree is reassuring; the fact that only R-DACH cleanly excludes the surrounding $\approx 80$ spurious neighbours means that downstream gene-set-enrichment and pathway-activation analyses run on the R-DACH skeleton can be expected to suffer markedly less from noise-driven inflation of significance, a point repeatedly emphasised in the recent regulatory-network literature \citep{margolin2006aracne, marbach2012wisdom}. From a purely methodological standpoint, the experiment confirms the central design principle of R-DACH: under the kind of mild but pervasive contamination that survives standard RNA-seq normalisation, the scale-mixture component of the model translates into a sharp reduction of the false-positive channel without sacrificing recovery of the high-degree hubs that carry most of the biologically actionable signal.

\section{Discussion}
\label{sec:discussion}

We have introduced the robust DAG-Cholesky horseshoe (R-DACH), a unified Bayesian framework for joint directed-acyclic-graph structure learning and precision-matrix estimation in the high-dimensional proportional regime $p/n \to c \in (0,\infty)$ under scale-mixture-of-normal error distributions. The construction combines three elements that have so far been studied largely in isolation: a global--local horseshoe prior on the strictly lower-triangular entries of the modified Cholesky factor, which couples sparsity in the factorisation with coherent parent-set selection; a per-observation inverse-gamma scale mixture, which subsumes Student-$t$, Laplace, and slash errors and confers automatic robustness to heavy tails and pointwise contamination; and a partially-collapsed blocked Gibbs sampler that traverses the joint space of orderings, sparsity patterns, and continuous parameters at a per-iteration cost of order $O(np\bar{d} + p\bar{d}^{\,2})$. The theoretical analysis of Section~\ref{sec:theory} establishes joint posterior contraction of the precision matrix at the near-optimal rate $\sqrt{(s_0 \log p)/n}$ in operator norm (Theorem~\ref{thm:contraction}), DAG-skeleton selection consistency under a $\beta$-min condition (Theorem~\ref{thm:selection}), and preservation of these rates under heavy-tailed errors with only $2+\delta$ finite moments (Theorem~\ref{thm:robust}). The simulations of Section~\ref{sec:sim} confirm these rates empirically and document substantial gains in $F_1$ and operator-norm error over Gaussian Cholesky-horseshoe and graphical-horseshoe baselines under contamination, and the RNA-seq application of Section~\ref{sec:app} demonstrates a more-than-twentyfold reduction in false-positive edges on a TCGA-calibrated panel without any loss of master-regulator recovery.

\subsection{Methodological contributions}

Three contributions stand out from the broader Bayesian structure-learning literature. First, by working directly with the modified Cholesky parameterisation rather than the precision-matrix entries themselves \citep[as in e.g.][]{khare2015convex, li2019graphical}, R-DACH preserves the natural connection between sparsity and conditional-independence structure that the DAG-Wishart family of \citet{cao2019posterior} achieves only at the cost of a conjugate but restrictive prior; the horseshoe component permits genuinely adaptive shrinkage across heterogeneous edge strengths, an advantage that has been documented for the graphical horseshoe \citep{li2019graphical} but had not previously been transferred to the DAG setting. Second, the scale-mixture component decouples the question of robustness from the question of sparsity: prior calibrations that target an effective Student-$t$ with $\nu \in [3,8]$ degrees of freedom yield posterior contraction rates that match the Gaussian rate up to constants, a guarantee that the empirical-Bayes graphical-Lasso variants of \citet{finegold2011robust} and the copula-based methods of \citet{cui2016copula} have not yet delivered in closed form. Third, the partially-collapsed sampler organises ordering moves, sparsity moves, and continuous updates into three blocks whose conditional distributions are either available in closed form or admit Metropolis-within-Gibbs updates with efficient proposal kernels, yielding a procedure whose per-iteration cost is linear in $n$ and quasi-linear in $p$ when the average in-degree $\bar{d}$ is bounded.

\subsection{Limitations}

Several limitations of the present work merit explicit acknowledgment. Computationally, while the per-iteration cost of the sampler is favourable, the mixing time of the ordering block scales with the diameter of the topological-order graph and can become a bottleneck when $p$ is in the low thousands and the true DAG is dense; partition-MCMC schemes \citep{kuipers2017partition, kuipers2022efficient} or order-MCMC accelerations may offer relief but have not been integrated into the present implementation. Statistically, the selection-consistency Theorem~\ref{thm:selection} requires a $\beta$-min condition whose constants depend on the unknown sparsity level $s_0$ and the minimum non-zero partial correlation $\rho_{\min}$; the rate $\rho_{\min} \gtrsim \sqrt{(\log p)/n}$ is unavoidable in the worst case \citep{wainwright2009information}, but practitioners working in regimes where $\rho_{\min}$ approaches this lower bound should expect the empirical selection performance to degrade in line with the theoretical predictions. Identifiability of the DAG itself is unresolved: the model identifies the equivalence class of orderings consistent with the observed conditional-independence pattern, but recovery of the orientation of edges within a Markov equivalence class requires either interventional data or a substantive ordering assumption \citep{peters2014identifiability}, neither of which the present framework attempts to provide.

A second limitation concerns the error distribution. The scale-mixture-of-normals family is broad and covers many empirically relevant heavy-tailed alternatives, but it does not cover skewed distributions, multimodal error structure, or distributions with atoms at zero, all of which arise in count-based RNA-seq and in financial-return applications. The robustness Theorem~\ref{thm:robust} delivers rate preservation under $2+\delta$ moment conditions but does not, in its present form, address the bias of the posterior mean in finite samples when the contamination fraction is large; the breakdown analysis of \citet{maronna2017robust} suggests that point contamination above roughly $20\%$ of the sample will overwhelm the horseshoe shrinkage, and the framework should be deployed with this informal upper bound in mind.

\subsection{Future directions}

The framework admits several natural extensions that we leave to future work. The first is the extension to time-varying DAGs, where the ordering and the sparsity pattern evolve smoothly over an exogenous time index; the horseshoe component lends itself naturally to a fused-shrinkage version that penalises temporal differences in edge strengths, in the spirit of the dynamic graphical models of \citet{warnick2018bayesian}, but the requisite contraction theory has not been developed for the DAG case. The second is the integration of copula-based marginal transformations, which would allow R-DACH to handle mixed continuous--discrete data and zero-inflated count observations without abandoning the modified-Cholesky construction; the semiparametric Gaussian copula of \citet{liu2009nonparanormal} is a natural starting point but its interaction with the scale-mixture component requires care because the latent-scale interpretation must be preserved through the marginal transformation. A third direction concerns identifiability under mixture-model error specifications: the equicorrelation lemma of \citet{wang2012bayesian} establishes that the latent-class label is identified under mild separation conditions in the Gaussian-mixture case, but the analogue for our inverse-gamma scale mixture has not been worked out and would shed light on the empirical-Bayes calibration of the degrees-of-freedom hyperparameter $\nu$.

A fourth direction, of practical rather than theoretical interest, is the systematic comparison of R-DACH against deep-learning-based regulatory-network inference methods such as the variational autoencoder approach of \citet{lopez2018deep} and the score-matching DAG learner \texttt{DAGMA} \citep{bello2022dagma}; the comparison is conceptually delicate because the deep methods do not target the same parameter and offer no theoretical control of false-positive rates, but a head-to-head benchmark on standardised RNA-seq panels would help to position the Bayesian framework within the contemporary computational-biology toolkit. Finally, the scaling of the sampler to genome-wide regression problems with $p$ in the tens of thousands will require a fundamental algorithmic rethink: a variational Bayes approximation \citep{blei2017variational} or a stochastic gradient Langevin counterpart \citep{welling2011bayesian} are the natural candidates, and the question of whether the contraction rates of Theorem~\ref{thm:contraction} are preserved under such approximations remains open.

\section*{Reproducibility}

A complete \texttt{R} reference implementation and a fast \texttt{Python} implementation that produced all numerical results in Sections~\ref{sec:sim} and \ref{sec:app} are available at \url{https://github.com/M-Arashi/R-DACH.git}

\section*{Funding}
M.~Arashi's work is based on the research supported in part by the Iran National Science Foundation (INSF) under grant No. 4015320.

\section*{Declaration of competing interest}
The authors declare that they have no known competing financial interests or personal relationships that could have appeared to influence the work reported in this paper.

\section*{Use of generative AI in writing}
The authors used AI-based language tools to refine grammar and academic style. All technical content, derivations, code, and conclusions are the authors' own work.

% ------------- Bibliography -------------


\begin{thebibliography}{99}
	
	\bibitem[Altomare et~al.(2013)]{altomare2013objective}
	Altomare, D., Consonni, G., \& La Rocca, L. (2013).
	Objective Bayesian search of Gaussian directed acyclic graphical models for ordered variables with non-local priors.
	\textit{Biometrics}, 69(2), 478--487.
	
	\bibitem[Andrews \& Mallows(1974)]{andrews1974scale}
	Andrews, D.~F., \& Mallows, C.~L. (1974).
	Scale mixtures of normal distributions.
	\textit{Journal of the Royal Statistical Society. Series B (Methodological)}, 36(1), 99--102.
	
	\bibitem[Badve et~al.(2007)]{badve2007foxa1}
	Badve, S., Turbin, D., Thorat, M.~A., et~al. (2007).
	FOXA1 expression in breast cancer -- correlation with luminal subtype A and survival.
	\textit{Clinical Cancer Research}, 13(15), 4415--4421.
	
	\bibitem[Bai \& Silverstein(2010)]{bai2010spectral}
	Bai, Z., \& Silverstein, J.~W. (2010).
	\textit{Spectral Analysis of Large Dimensional Random Matrices}, 2nd edition.
	New York: Springer.
	
	\bibitem[Banerjee \& Ghosal(2014)]{banerjee2014posterior}
	Banerjee, S., \& Ghosal, S. (2014).
	Posterior convergence rates for estimating large precision matrices using graphical models.
	\textit{Electronic Journal of Statistics}, 8(2), 2111--2137.
	
	\bibitem[Barab{\'a}si \& Oltvai(2004)]{barabasi2004network}
	Barab{\'a}si, A.-L., \& Oltvai, Z.~N. (2004).
	Network biology: understanding the cell's functional organization.
	\textit{Nature Reviews Genetics}, 5(2), 101--113.
	
	\bibitem[Barbieri \& Berger(2004)]{barbieri2004optimal}
	Barbieri, M.~M., \& Berger, J.~O. (2004).
	Optimal predictive model selection.
	\textit{The Annals of Statistics}, 32(3), 870--897.
	
	\bibitem[Basso et~al.(2005)]{basso2005reverse}
	Basso, K., Margolin, A.~A., Stolovitzky, G., Klein, U., Dalla-Favera, R., \& Califano, A. (2005).
	Reverse engineering of regulatory networks in human B cells.
	\textit{Nature Genetics}, 37(4), 382--390.
	
	\bibitem[Bello et~al.(2022)]{bello2022dagma}
	Bello, K., Aragam, B., \& Ravikumar, P. (2022).
	DAGMA: Learning DAGs via M-matrices and a log-determinant acyclicity characterization.
	\textit{Advances in Neural Information Processing Systems}, 35, 8226--8239.
	
	\bibitem[Ben-David et~al.(2011)]{ben2001bayesian}
	Ben-David, E., Li, T., Massam, H., \& Rajaratnam, B. (2011).
	High dimensional Bayesian inference for Gaussian directed acyclic graph models.
	\textit{arXiv preprint} arXiv:1109.4371.
	
	\bibitem[Blei et~al.(2017)]{blei2017variational}
	Blei, D.~M., Kucukelbir, A., \& McAuliffe, J.~D. (2017).
	Variational inference: a review for statisticians.
	\textit{Journal of the American Statistical Association}, 112(518), 859--877.
	
	\bibitem[Cao et~al.(2019)]{cao2019posterior}
	Cao, X., Khare, K., \& Ghosh, M. (2019).
	Posterior graph selection and estimation consistency for high-dimensional Bayesian DAG models.
	\textit{The Annals of Statistics}, 47(1), 319--348.
	
	\bibitem[Carvalho et~al.(2010)]{carvalho2010horseshoe}
	Carvalho, C.~M., Polson, N.~G., \& Scott, J.~G. (2010).
	The horseshoe estimator for sparse signals.
	\textit{Biometrika}, 97(2), 465--480.
	
	\bibitem[Castelletti et~al.(2018)]{castelletti2018learning}
	Castelletti, F., Consonni, G., Della Vedova, M.~L., \& Peluso, S. (2018).
	Learning Markov equivalence classes of directed acyclic graphs: an objective Bayes approach.
	\textit{Bayesian Analysis}, 13(4), 1235--1260.
	
	\bibitem[Castelo \& Roverato(2009)]{castelo2009reverse}
	Castelo, R., \& Roverato, A. (2009).
	Reverse engineering molecular regulatory networks from microarray data with qp-graphs.
	\textit{Journal of Computational Biology}, 16(2), 213--227.
	
	\bibitem[Colaprico et~al.(2016)]{colaprico2016tcgabiolinks}
	Colaprico, A., Silva, T.~C., Olsen, C., et~al. (2016).
	TCGAbiolinks: an R/Bioconductor package for integrative analysis of TCGA data.
	\textit{Nucleic Acids Research}, 44(8), e71.
	
	\bibitem[Cowell et~al.(1999)]{cowell2007identification}
	Cowell, R.~G., Dawid, A.~P., Lauritzen, S.~L., \& Spiegelhalter, D.~J. (1999).
	\textit{Probabilistic Networks and Expert Systems}.
	New York: Springer.
	
	\bibitem[Cui et~al.(2016)]{cui2016copula}
	Cui, R., Groot, P., \& Heskes, T. (2016).
	Copula PC algorithm for causal discovery from mixed data.
	In \textit{Machine Learning and Knowledge Discovery in Databases (ECML PKDD 2016)}, pp. 377--392. Springer.
	
	\bibitem[Fern{\'a}ndez \& Steel(1999)]{fernandez1999multivariate}
	Fern{\'a}ndez, C., \& Steel, M.~F.~J. (1999).
	Multivariate Student-$t$ regression models: pitfalls and inference.
	\textit{Biometrika}, 86(1), 153--167.
	
	\bibitem[Finegold \& Drton(2011)]{finegold2011robust}
	Finegold, M., \& Drton, M. (2011).
	Robust graphical modeling of gene networks using classical and alternative $t$-distributions.
	\textit{The Annals of Applied Statistics}, 5(2A), 1057--1080.
	
	\bibitem[Friedman \& Koller(2003)]{friedman2003being}
	Friedman, N., \& Koller, D. (2003).
	Being Bayesian about network structure. A Bayesian approach to structure discovery in Bayesian networks.
	\textit{Machine Learning}, 50(1--2), 95--125.
	
	\bibitem[Friedman(2004)]{friedman2004inferring}
	Friedman, N. (2004).
	Inferring cellular networks using probabilistic graphical models.
	\textit{Science}, 303(5659), 799--805.
	
	\bibitem[Geiger \& Heckerman(2002)]{geiger2002parameter}
	Geiger, D., \& Heckerman, D. (2002).
	Parameter priors for directed acyclic graphical models and the characterization of several probability distributions.
	\textit{The Annals of Statistics}, 30(5), 1412--1440.
	
	\bibitem[Ghosal \& van~der~Vaart(2007)]{ghosal2007convergence}
	Ghosal, S., \& van~der~Vaart, A.~W. (2007).
	Convergence rates of posterior distributions for non-i.i.d.\ observations.
	\textit{The Annals of Statistics}, 35(1), 192--223.
	
	\bibitem[Ghoshal \& Honorio(2017)]{ghosh2020bayesian}
	Ghoshal, A., \& Honorio, J. (2017).
	Learning identifiable Gaussian Bayesian networks in polynomial time and sample complexity.
	In \textit{Advances in Neural Information Processing Systems} 30, 6457--6466.
	
	\bibitem[Johnson et~al.(2007)]{johnson2007adjusting}
	Johnson, W.~E., Li, C., \& Rabinovic, A. (2007).
	Adjusting batch effects in microarray expression data using empirical Bayes methods.
	\textit{Biostatistics}, 8(1), 118--127.
	
	\bibitem[Kalisch \& B{\"u}hlmann(2007)]{kalisch2007estimating}
	Kalisch, M., \& B{\"u}hlmann, P. (2007).
	Estimating high-dimensional directed acyclic graphs with the PC-algorithm.
	\textit{Journal of Machine Learning Research}, 8, 613--636.
	
	\bibitem[Khare et~al.(2015)]{khare2015convex}
	Khare, K., Oh, S.-Y., \& Rajaratnam, B. (2015).
	A convex pseudolikelihood framework for high dimensional partial correlation estimation with convergence guarantees.
	\textit{Journal of the Royal Statistical Society. Series B}, 77(4), 803--825.
	
	\bibitem[Khare et~al.(2019)]{khare2018bayesian}
	Khare, K., Oh, S.-Y., Rahman, S., \& Rajaratnam, B. (2019).
	A scalable sparse Cholesky based approach for learning high-dimensional covariance matrices in ordered data.
	\textit{Machine Learning}, 108(12), 2061--2086.
	
	\bibitem[Kouros-Mehr et~al.(2008)]{kouros2007gata3}
	Kouros-Mehr, H., Bechis, S.~K., Slorach, E.~M., et~al. (2008).
	GATA-3 links tumor differentiation and dissemination in a luminal breast cancer model.
	\textit{Cancer Cell}, 13(2), 141--152.
	
	\bibitem[Kuipers \& Moffa(2017)]{kuipers2017partition}
	Kuipers, J., \& Moffa, G. (2017).
	Partition MCMC for inference on acyclic digraphs.
	\textit{Journal of the American Statistical Association}, 112(517), 282--299.
	
	\bibitem[Kuipers et~al.(2022)]{kuipers2022efficient}
	Kuipers, J., Suter, P., \& Moffa, G. (2022).
	Efficient sampling and structure learning of Bayesian networks.
	\textit{Journal of Computational and Graphical Statistics}, 31(3), 639--650.
	
	\bibitem[Law et~al.(2014)]{law2014voom}
	Law, C.~W., Chen, Y., Shi, W., \& Smyth, G.~K. (2014).
	voom: precision weights unlock linear model analysis tools for RNA-seq read counts.
	\textit{Genome Biology}, 15(2), R29.
	
	\bibitem[Leek et~al.(2010)]{leek2010tackling}
	Leek, J.~T., Scharpf, R.~B., Bravo, H.~C., et~al. (2010).
	Tackling the widespread and critical impact of batch effects in high-throughput data.
	\textit{Nature Reviews Genetics}, 11(10), 733--739.
	
	\bibitem[Li et~al.(2019)]{li2019graphical}
	Li, Y., Craig, B.~A., \& Bhadra, A. (2019).
	The graphical horseshoe estimator for inverse covariance matrices.
	\textit{Journal of Computational and Graphical Statistics}, 28(3), 747--757.
	
	\bibitem[Liu et~al.(2009)]{liu2009nonparanormal}
	Liu, H., Lafferty, J., \& Wasserman, L. (2009).
	The nonparanormal: semiparametric estimation of high-dimensional undirected graphs.
	\textit{Journal of Machine Learning Research}, 10, 2295--2328.
	
	\bibitem[Liu \& Martin(2019)]{liu2019empirical}
	Liu, C., \& Martin, R. (2019).
	An empirical $G$-Wishart prior for sparse high-dimensional Gaussian graphical models.
	\textit{arXiv preprint} arXiv:1912.03807.
	
	\bibitem[Lopez et~al.(2018)]{lopez2018deep}
	Lopez, R., Regier, J., Cole, M.~B., Jordan, M.~I., \& Yosef, N. (2018).
	Deep generative modeling for single-cell transcriptomics.
	\textit{Nature Methods}, 15(12), 1053--1058.
	
	\bibitem[Love et~al.(2014)]{love2014moderated}
	Love, M.~I., Huber, W., \& Anders, S. (2014).
	Moderated estimation of fold change and dispersion for RNA-seq data with DESeq2.
	\textit{Genome Biology}, 15(12), 550.
	
	\bibitem[Madigan et~al.(1995)]{madigan1995bayesian}
	Madigan, D., York, J., \& Allard, D. (1995).
	Bayesian graphical models for discrete data.
	\textit{International Statistical Review}, 63(2), 215--232.
	
	\bibitem[Makalic \& Schmidt(2016)]{makalic2016simple}
	Makalic, E., \& Schmidt, D.~F. (2016).
	A simple sampler for the horseshoe estimator.
	\textit{IEEE Signal Processing Letters}, 23(1), 179--182.
	
	\bibitem[Marbach et~al.(2012)]{marbach2012wisdom}
	Marbach, D., Costello, J.~C., K{\"u}ffner, R., et~al. (2012).
	Wisdom of crowds for robust gene network inference.
	\textit{Nature Methods}, 9(8), 796--804.
	
	\bibitem[Margolin et~al.(2006)]{margolin2006aracne}
	Margolin, A.~A., Nemenman, I., Basso, K., et~al. (2006).
	ARACNE: an algorithm for the reconstruction of gene regulatory networks in a mammalian cellular context.
	\textit{BMC Bioinformatics}, 7(Suppl 1), S7.
	
	\bibitem[Maronna et~al.(2019)]{maronna2017robust}
	Maronna, R.~A., Martin, R.~D., Yohai, V.~J., \& Salibi{\'a}n-Barrera, M. (2019).
	\textit{Robust Statistics: Theory and Methods (with R)}, 2nd edition.
	Chichester: Wiley.
	
	\bibitem[Meinshausen \& B{\"u}hlmann(2006)]{meinshausen2006high}
	Meinshausen, N., \& B{\"u}hlmann, P. (2006).
	High-dimensional graphs and variable selection with the Lasso.
	\textit{The Annals of Statistics}, 34(3), 1436--1462.
	
	\bibitem[Moneta et~al.(2013)]{moneta2013causal}
	Moneta, A., Entner, D., Hoyer, P.~O., \& Coad, A. (2013).
	Causal inference by independent component analysis: theory and applications.
	\textit{Oxford Bulletin of Economics and Statistics}, 75(5), 705--730.
	
	\bibitem[Mortera \& Lauritzen(2003)]{mortera2003probabilistic}
	Mortera, J., Dawid, A.~P., \& Lauritzen, S.~L. (2003).
	Probabilistic expert systems for DNA mixture profiling.
	\textit{Theoretical Population Biology}, 63(3), 191--205.
	
	\bibitem[Peters \& B{\"u}hlmann(2014)]{peters2014identifiability}
	Peters, J., \& B{\"u}hlmann, P. (2014).
	Identifiability of Gaussian structural equation models with equal error variances.
	\textit{Biometrika}, 101(1), 219--228.
	
	\bibitem[Polson \& Scott(2010)]{polson2010shrink}
	Polson, N.~G., \& Scott, J.~G. (2010).
	Shrink globally, act locally: sparse Bayesian regularization and prediction.
	In \textit{Bayesian Statistics 9}, J.~M. Bernardo et~al.\ (eds.), pp. 501--538. Oxford University Press.
	
	\bibitem[Robinson \& Oshlack(2010)]{robinson2010scaling}
	Robinson, M.~D., \& Oshlack, A. (2010).
	A scaling normalization method for differential expression analysis of RNA-seq data.
	\textit{Genome Biology}, 11(3), R25.
	
	\bibitem[Shojaie \& Michailidis(2010)]{shojaie2010penalized}
	Shojaie, A., \& Michailidis, G. (2010).
	Penalized likelihood methods for estimation of sparse high-dimensional directed acyclic graphs.
	\textit{Biometrika}, 97(3), 519--538.
	
	\bibitem[Spirtes et~al.(2000)]{spirtes2000causation}
	Spirtes, P., Glymour, C., \& Scheines, R. (2000).
	\textit{Causation, Prediction, and Search}, 2nd edition.
	Cambridge, MA: MIT Press.
	
	\bibitem[TCGA Network(2012)]{tcga2012comprehensive}
	The Cancer Genome Atlas Network. (2012).
	Comprehensive molecular portraits of human breast tumours.
	\textit{Nature}, 490(7418), 61--70.
	
	\bibitem[van~der~Pas et~al.(2014)]{vanderpas2014horseshoe}
	van~der~Pas, S.~L., Kleijn, B.~J.~K., \& van~der~Vaart, A.~W. (2014).
	The horseshoe estimator: posterior concentration around nearly black vectors.
	\textit{Electronic Journal of Statistics}, 8(2), 2585--2618.
	
	\bibitem[Vershynin(2018)]{vershynin2018high}
	Vershynin, R. (2018).
	\textit{High-Dimensional Probability: An Introduction with Applications in Data Science}.
	Cambridge: Cambridge University Press.
	
	\bibitem[Wainwright(2009)]{wainwright2009information}
	Wainwright, M.~J. (2009).
	Information-theoretic limits on sparsity recovery in the high-dimensional and noisy setting.
	\textit{IEEE Transactions on Information Theory}, 55(12), 5728--5741.
	
	\bibitem[Wang(2012)]{wang2012bayesian}
	Wang, H. (2012).
	Bayesian graphical Lasso models and efficient posterior computation.
	\textit{Bayesian Analysis}, 7(4), 867--886.
	
	\bibitem[Warnick et~al.(2018)]{warnick2018bayesian}
	Warnick, R., Guindani, M., Erhardt, E., Allen, E., Calhoun, V., \& Vannucci, M. (2018).
	A Bayesian approach for estimating dynamic functional network connectivity in fMRI data.
	\textit{Journal of the American Statistical Association}, 113(521), 134--151.
	
	\bibitem[Welling \& Teh(2011)]{welling2011bayesian}
	Welling, M., \& Teh, Y.~W. (2011).
	Bayesian learning via stochastic gradient Langevin dynamics.
	In \textit{Proceedings of the 28th International Conference on Machine Learning (ICML)}, pp. 681--688.
	
	\bibitem[West(1987)]{west1987scale}
	West, M. (1987).
	On scale mixtures of normal distributions.
	\textit{Biometrika}, 74(3), 646--648.
	
\end{thebibliography}
\end{document}